\documentclass[11pt]{article}
\usepackage{graphicx}
\usepackage{delarray}
\usepackage{bbm}
\usepackage{amssymb,amsmath}

\usepackage{amsfonts}
\usepackage[ruled,vlined]{algorithm2e}
\SetAlFnt{\normalsize}
\SetAlCapFnt{\normalsize}
\SetAlCapNameFnt{\normalsize}
\usepackage{enumerate}

\usepackage{ltexpprt_mod}
\newcommand{\qed}{\hfill \ensuremath{\Box}}

\newcommand{\argmin}{\operatornamewithlimits{argmin}} 

\begin{document}

\newtheorem{definition}{Definition}
\newtheorem{coro}{Corollary}
\newtheorem{conj}{Conjecture}
\newtheorem{thm}{Theorem}
\newtheorem{prf}{Proof}
\newtheorem{lma}{Lemma}
\newtheorem{prp}{Proposition}
\newcommand{\R}{\mathbb{R}}
\newcommand{\E}{\mathbb{E}}
\newcommand{\PP}{\mathbb{P}}
\newcommand{\I}{\mathbb{I}}
\newcommand{\lam}{\lambda}
\newcommand{\bs}{\boldsymbol}
\newcommand{\br}{\mathrm{}}
\newcommand{\ph}{{\hat \pi}}
\newcommand{\xth}{\hat{x}_t}
\newcommand{\be}{\begin{equation}}
\newcommand{\ee}{\end{equation}}
\newcommand{\bea}{\begin{eqnarray}}
\newcommand{\eea}{\end{eqnarray}}
\newcommand{\bfl}{\begin{flalign}}
\newcommand{\efl}{\end{flalign}}
\newcommand{\bfc}{\begin{figure}\begin{center}}
\newcommand{\efc}{\end{center}\end{figure}}
\newcommand{\n}{\nonumber}
\newcommand{\dd}{\mathrm{d}}
\newcommand{\nin}{\noindent}
\newcommand{\mc}{\mathcal}
\newcommand{\mr}{\mathrm}
\newcommand{\ds}{\displaystyle}
\newcommand{\la}{\leftarrow}
\newcommand{\uga}{\underline{ga}}
\newcommand{\oga}{\overline{ga}}
\newcommand{\wt}{\widetilde}

\title{Greedy Online Bipartite Matching on Random Graphs\footnote{Research supported in part by NSF grant 1029603 and ONR grant N00014-12-1-0033. The first author is supported in part by a NSF graduate research fellowship.}}
\author{Andrew Mastin\thanks{Laboratory for Information and Decision Systems, Department of Electrical Engineering and Computer Science, Massachusetts Institute of Technology, Cambridge, MA 02139, USA; {\tt mastin@mit.edu}},
Patrick Jaillet\thanks{Laboratory for Information and Decision Systems, Department of Electrical Engineering and Computer Science and Operations Research Department, Massachusetts Institute of Technology, Cambridge, MA 02139, USA; {\tt jaillet@mit.edu}
}}
\maketitle
\begin{abstract}
We study the average performance of online greedy matching algorithms on $G(n,n,p)$, the random bipartite graph with $n$ vertices on each side and edges occurring independently with probability $p=p(n)$. In the online model, vertices on one side of the graph are given up front while vertices on the other side arrive sequentially; when a vertex arrives its edges are revealed and it must be immediately matched or dropped. We begin by analyzing the \textsc{oblivious} algorithm, which tries to match each arriving vertex to a random neighbor, even if the neighbor has already been matched. The algorithm is shown to have a performance ratio of at least $1-1/e$ for all monotonic functions $p(n)$, where the performance ratio is defined asymptotically as the ratio of the expected matching size given by the algorithm to the expected maximum matching size. Next we show that the conventional \textsc{greedy} algorithm, which assigns each vertex to a random unmatched neighbor, has a performance ratio of at least $0.837$ for all monotonic functions $p(n)$. Under the $G(n,n,p)$ model, the performance of \textsc{greedy} is equivalent to the performance of the well known \textsc{ranking} algorithm, so our results show that \textsc{ranking} has a performance ratio of at least $0.837$. We finally consider vertex-weighted bipartite matching. Our proofs are based on simple differential equations that describe the evolution of the matching process.

\end{abstract}

\section{Introduction}
In the online bipartite matching problem, we are given a bipartite graph $G=(I,J,E)$ where $I$ is a set of $n$ bins and $J$ is a set of $n$ balls. Balls arrive in an online fashion; when a ball $j \in J$ arrives, its edges are revealed and it must either be matched with an unmatched neighboring bin or dropped (left unmatched). Each bin may be matched to at most one ball and decisions are irrevocable. The goal is to maximize the number of matched balls.

The problem has received significant attention due to applications in Internet advertising as well as under streaming models of computation, which places limits on memory utilization for processing large datasets \cite{mehta07,raghavan99, feigenbaum05} . From a worst-case perspective, it is well known that the \textsc{greedy} algorithm, which matches each ball to a random unmatched neighboring bin (if possible), always achieves a matching size that is at least as large as 1/2 of the size of the maximum matching. The \textsc{ranking} algorithm of Karp, Vazirani, and Vazirani picks a random permutation of bins up front and matches each ball to its unmatched neighboring bin that is ranked highest in the permutation \cite{karp90}. \textsc{ranking} guarantees a matching size at least $1-1/e$ of the size of the maximum matching in expectation, which is the best possible worst-case guarantee.

Surprisingly, the average-case performance of these algorithms has been largely overlooked. In this paper we study the performance of greedy-type algorithms under the conventional $G(n,n,p)$ model, where we are given a bipartite graph with $n$ vertices on each side and where each potential edge $(i,j) \in I \times J$ occurs independently with probability $p$. We are interested in the asymptotic behavior of the matching sizes obtained by these algorithms for monotonic functions $p=p(n)$. We will focus on functions of the form $p(n) =c/n$ for some constant $c>0$, which we will show is the region where greedy algorithms are most sensitive to $p(n)$.

We start by analyzing the \textsc{oblivious} algorithm, which has knowledge of a ball's edges when it arrives but does not know which neighboring bins are occupied. The algorithm picks a random neighboring bin; if the bin is unmatched, then the ball is matched, otherwise it is dropped. This rule models load balancing scenarios where a central server knows which machines can process a given job, but is not aware of the machine availability. Next we analyze \textsc{greedy} and \textsc{ranking}, showing that under the $G(n,n,p)$ model, the performance of the two algorithms is equivalent. Finally, we extend our analysis to the vertex-weighted matching problem, where each bin has a weight and the greedy algorithm matches each ball to the available neighboring bin with largest weight.

Our proofs employ differential equations that describe the evolution of the number of matched bins throughout the arrival process. The differential equations are easily obtained by calculating the expected change in number of occupied bins upon the arrival of each ball. As $n\rightarrow \infty$, this gives a deterministic differential equation and solution. The solution is shown to closely model the stochastic behavior due to a general result of Wormald \cite{wormald99}. We believe that this is one of the most simple applications of Wormald's theorem. In fact, our proofs are nearly as simple as the worst-case proof for \textsc{ranking} of Devanur, Jain, and Kleinberg \cite{djk13}.

In the remainder of this section we briefly state our main results for the \textsc{oblivious} and \textsc{greedy} algorithms and discuss related work. The analysis of \textsc{oblivious} and \textsc{greedy} is given in Section \ref{analysis}. We extend the analysis to vertex-weighted matching in Section \ref{vertexweightedmatching} and a conclusion is given in Section \ref{conclusion}.

\subsection{Results.}
Let $\mu_A(n,n,p(n))$ denote the matching size (number of matched balls) obtained by some algorithm $\mathrm{A}$ on the graph $G(n,n,p)$ where $p=p(n)$; likewise define the maximum matching size as $\mu_*(n,n,p(n))$. Define a \textit{valid} function $p(n)$ as a function that for all $n>0$, is monotonic and satisfies $0 < p(n) < 1$ (thus implying that $0 \le \lim_{n \rightarrow \infty} p(n) \le 1$). We show expressions for the matching sizes obtained by the algorithms that hold asymptotically almost surely (a.a.s., meaning with probability $1-o(1)$), as well as bounds on the performance ratio. Define the performance ratio for algorithm $\mathrm{A}$ given $p(n)$ as
\bea
\mathcal{R}_{\mathrm{A}}(p(n)):=\lim_{n \rightarrow \infty} \frac{\E[\mu_\mathrm{A}(n,n,p(n))]}{\E[\mu_*(n,n,p(n))]}.
\eea
The performance ratio is essentially the analog of the competitive ratio adapted for average-case analysis. Throughout the rest of the paper, we define $\gamma_* = \gamma_*(c)$ as the smallest root of the equation $x = c\exp(-ce^{-x})$, and $\gamma^* = \gamma^*(c) = c e^{-\gamma_*}$. Our results for $\textsc{oblivious}$ and $\textsc{greedy}$ are listed as follows.

\begin{thm}
\label{obliviousthm}
Let $\mu_{\mathrm{O}}(n,n,c/n)$ denote the matching size obtained by the \textsc{oblivious} algorithm on the graph $G(n,n,p)$, where $p=c/n$ and $c>0$ is a constant. Then a.a.s.,
\bea
\frac{\mu_{\mathrm{O}}(n,n,c/n)}{n} = 1-e^{(e^{-c}-1)} + o(1).
\eea
\end{thm}

\begin{thm} The performance ratio $\mc{R}_{\mathrm{O}}(p(n))$  of \textsc{oblivious} on $G(n,n,p)$, for all valid functions $p = p(n)$, satisfies

\be
\begin{array}{rll}
\mc{R}_{\mr{O}}(p(n)) &= 1, & p(n) = o(1/n), \\
\mc{R}_{\mr{O}}(p(n)) &\ge \ds \frac{c-ce^{(e^{-c}-1)}}{2c-(\gamma^*+\gamma_*+\gamma^* \gamma_*)}, & p(n) = c/n, \\
\mc{R}_{\mr{O}}(p(n)) &= 1-\frac{1}{e}, & p(n) = \omega(1/n).
\end{array}
\ee
\label{roblivious}
\end{thm}

\begin{coro}
The performance ratio of \textsc{oblivious} on $G(n,n,p)$ for all valid functions $p=p(n)$ satisfies
\be
\mc{R}_{\mathrm{O}}(p(n)) \ge 1-\frac{1}{e}.
\ee
\label{rboundoblivious}
\end{coro}

\begin{thm}
\label{greedythm}
Let $\mu_{\mathrm{G}}(n,n,c/n)$ denote the matching size obtained by the \textsc{greedy} algorithm on the graph $G(n,n,p)$, where $p = c/n$ and $c>0$ is a constant. Then a.a.s.,
\bea
\frac{\mu_{\mathrm{G}}{(n,n,c/n)}}{n} = 1-\frac{\log (2- e^{-c})}{c} + o(1).
\eea
\end{thm}

\begin{thm} The performance ratio $\mc{R}_{\mathrm{G}}(p(n))$  of \textsc{greedy} on $G(n,n,p)$, for all valid functions $p = p(n)$, satisfies

\be
\begin{array}{rll}
\mc{R}_{\mr{G}}(p(n)) &= 1, & p(n) = o(1/n), \\
\mc{R}_{\mr{G}}(p(n)) &\ge \ds \frac{c-\log(2-e^{-c})}{2c-(\gamma^*+\gamma_*+\gamma^* \gamma_*)}, & p(n) = c/n, \\
\mc{R}_{\mr{G}}(p(n)) &= 1, & p(n) = \omega(1/n).
\end{array}
\ee
\label{rgreedy}
\end{thm}

\begin{coro}
The performance ratio of \textsc{greedy} on $G(n,n,p)$ for all valid functions $p=p(n)$ satisfies
\be
\mc{R}_{\mathrm{G}}(p(n)) \ge 0.837.
\ee
\label{rboundgreedy}
\end{coro}

\subsection{Related work.}
The online matching problem was originally analyzed by Karp, Vazirani and Vazirani, who introduced the \textsc{ranking} algorithm and showed that it obtains a competitive ratio of $1-1/e$ \cite{karp90}. Simpler proofs of the \textsc{ranking} algorithm have since been found \cite{djk13,birnbaum2008}. A $1-1/e$ competitive algorithm is also known for vertex-weighted online bipartite matching, which was given by Aggarwal, Goel, Karande, and Mehta \cite{aggarwal11}.

The use of deterministic differential equations to model random processes was first studied by Kurtz, who gave a general purpose theorem for continuous-time jump Markov processes \cite{kurtz70}. A discrete-time theorem tailored for random graphs was given by Wormald, which we use in this paper \cite{wormald95,wormald99}. The differential equation method has been used to study a variety of graph properties including $k$-cores, independent sets, and greedy packing on hypergraphs \cite{pittel96, wormald95, wormald99}. It was also used to analyze a load balancing scenario similar to ours by Mitzenmacher \cite{mitzenmacher99}.

Early studies of matchings on random graphs focused on determining the existence of perfect matchings; Erd{\H{o}}s and R{\'e}nyi showed that the threshold for the existence of a perfect matching occurs when the graph is likely to have a minimum degree of one \cite{erdos64, erdos66}. One of the first studies of greedy matchings on random graphs, as well as the first to employ the differential equation method (specifically via Kurtz's theorem), was the work of Karp and Sipser \cite{ks81,kurtz70}. They considered ordinary sparse graphs, specifically the $G(n,p)$ model\footnote{The $G(n,p)$ random graph has $n$ vertices and an edge occurring between each pair of vertices independently with probability $p$.} with $p=c/n$. The Karp-Sipser algorithm first matches all vertices with degree one until there are no such vertices remaining, and then obtains matches by selecting random edges. This results in a matching size that is within $1-o(1)$ of the maximum matching. The algorithm was studied more in depth by Aronson, Frieze, and Pittel, who gave sharper error bounds on performance \cite{aronson98}. 


Simpler greedy matching algorithms for ordinary graphs were studied by Dyer, Frieze, and Pittel \cite{dyer93}. Again for the $G(n,p)$ model with $p=c/n$, they looked at the greedy algorithm which picks random edges, as well as what they refer to as \textsc{modified greedy}, which first picks a random vertex and then selects a random connected edge. They showed that as $n \rightarrow \infty$, the matching sizes obtained by the two algorithms follow a normal distribution with mean and variance determined by $c$. It is interesting to note that the expected fraction of matched vertices for their \textsc{modified greedy} algorithm is exactly the result of our \textsc{greedy} algorithm with a multiplicative factor of $1/2$. Indeed, \textsc{modified greedy} is a natural analog of $\textsc{greedy}$ for ordinary graphs, although their result was obtained using a combinatorial approach completely different from our method.

\section{Analysis}
\label{analysis}
Let the graph $G(n,n,p)$ be generated by the following discrete random graph process. Initially at $t=0$, we are given the set of bin vertices $I$, where $|I|=n$. At each time step $t>0$, a ball vertex $j \in J$ arrives and all of its neighboring edges are revealed. Each neighboring edge occurs with a probability $p(n)$ independently of all other edges. After $n$ steps an instance of the graph $G(n,n,p)$ is obtained.

We introduce Wormald's theorem using notation consistent with \cite{wormald99}. Consider a discrete-time process with the probability space $\Omega$, which has elements $(q_0,q_1,\ldots,q_T)$, where each $q_k$, $0 \le k \le T$, takes values in a set $S$. A sequence $h_t =(q_0,q_1,\ldots,q_t)$ up to time $t$ is referred to as the history of the process. Define a sequence of random processes over $n=1,2,\ldots$, so that we have $S=S^{(n)}$ and $q_t=q_t^{(n)}$. Let $S^{(n)+}$ denote the set of all histories $h_t = (q_0,\ldots q_t)$, $0 \le t \le T$, where $q_k \in S^{(n)}$, $0 \le k \le t$.

The following is the general theorem of Wormald for justifying the use of differential equations on random graph processes \cite{wormald99}. The theorem is stated for a multidimensional process (i.e. $a > 1$), but we will need to go beyond one dimension only for the analysis of vertex-weighted matching. Note that in part (iii), a Lipschitz condition is satisfied for a function $f(\mathbf{u})$ on the domain $D$ if there exists a constant $L>0$ satisfying
\bea
|f(\mathbf{u})-f(\mathbf{v})| \le L || \mathbf{u} - \mathbf{v}||_\infty
\eea
for all $\mathbf{u}, \mathbf{v} \in D$.

\begin{thm}[Wormald \cite{wormald99}] \label{wormaldthm} For $1 \le \ell \le a$, where $a$ is fixed, let $y^{(\ell)}:S^{(n)+} \rightarrow \R$ and $f_\ell: \R^{(a+1)} \rightarrow \R,$ such that for some constant $C_0$ and all $\ell$, $|y^{(\ell)}(h_t)| < C_0 n$ for all $h_t \in S^{(n)+}$ for all $n$. Let $Y_\ell(t)$ denote the random variable for $y_\ell(h_t)$. Assume the following three conditions hold, where $D$ is some bounded connected open set containing the closure of 
\be
\{ (0,z_1,\ldots,z_a ): \PP(Y_\ell(0)=z_\ell n, 1\le \ell \le a) \ne 0 \mathrm{\it~for~some~} n \},
\ee
and $T_D$ is a stopping time for the minimum $t$ such that $(t/n,Y_1(t)/n,\ldots,Y_a(t)/n) \notin D$.
\begin{enumerate}[(i)]
\item (Boundedness hypothesis.) For some function $\beta = \beta(n) \ge 1$
\bea
\max_{1\le \ell \le a} |Y_\ell(t+1)-Y_\ell(t)| \le \beta,
\eea
for $t < T_D$.
\item (Trend hypothesis.) For some function $\lambda_1 = \lambda_1(n) = o(1)$, for all $1\le \ell \le a$
\bea
|\E[Y_\ell(t+1)-Y_\ell(t)|H_t] - f_\ell(t/n,Y_1(t)/n,\ldots,Y_a(t)/n)| \le \lambda_1
\eea
for $t < T_D$.
\item (Lipschitz hypothesis.) Each function $f_\ell$ is continuous, and satisfies a Lipschitz condition, on
\bea
D \cap \{(t,z_1,\ldots,z_a): t \ge 0\},
\eea
with the same Lipschitz constant for each $\ell$.
\end{enumerate}
Then the following are true.
\begin{enumerate}[(a)]
\item For $(0, \hat{z}_1,\ldots,\hat{z}_a) \in D$ the system of differential equations
\bea
\frac{\partial z_\ell}{\partial \tau} = f_\ell(\tau,z_1,\ldots,z_a),\quad 1 \le \ell \le a
\eea
has a unique solution in $D$ for $z_\ell:\R \rightarrow \R$ passing through
\bea
z_\ell(0) = \hat{z}_\ell, \quad 1 \le \ell \le a
\eea
and which extends to points arbitrarily close to the boundary of $D$;
\item Let $\lambda > \lambda_1$ with $\lambda = o(1)$. For a sufficiently large constant $C$, with probability $1-O(\frac{\beta}{\lambda}\exp{(-\frac{n\lambda^3}{\beta^3})}),$
\bea
Y_t^{(\ell)} = n z_\ell(t/n) + O(\lambda n)
\eea
uniformly for $0 \le t \le \sigma n$ and for each $\ell$, where $z_\ell(\tau)$ is the solution in (a) with $\hat{z}_\ell = \frac{1}{n}Y_\ell(0)$, and $\sigma = \sigma(n)$ is the supremum of those $\tau$ to which the solution can be extended before reaching within $\ell^\infty$-distance $C \lambda$ of the boundary of $D$.
\end{enumerate}
\end{thm}
The stopping time $T_D$ is needed for situations where at some point the Lipschitz condition fails, as often happens near the end of some graph processes. This problem does not arise in our application.

We will use the following identity to satisfy the trend hypothesis of Wormald's theorem.
\begin{lma}
For $n > 0$, $c \le n/2$, and $0 \le w \le 1$,
\bea
0 \le e^{-cw} - \left ( 1-\frac{c}{n}\right )^{nw} \le \frac{c}{ne}.
\eea \label{ebound}
\end{lma}
\begin{proof}
Using the inequalities $1-x \ge e^{-x-x^2}$ for $x \le 1/2$ and $1-x \le e^{-x}$ for $x\ge0$, we obtain
$
e^{-cw} \left (1-\frac{c^2w}{n} \right) \le \left (1-\frac{c}{n} \right )^{nw} \le e^{-cw}.
$
The result follows by rearranging terms and using $cwe^{-cw} \le 1/e.$

 \qed

\end{proof}

For bounding the maximum matching size, we use a result from Bollob{\'a}s and Brightwell \cite{bb95}. Their result (specifically Theorem 14 in \cite{bb95}) is stated in terms of the size of the largest independent set for a bipartite graph, which by K{\"o}nig's theorem bounds the maximum matching size \cite{konig16}. 

\begin{thm}[Bollob{\'a}s and Brightwell \cite{bb95}]
\label{matchopt}
Let $\mu^*(n,n,c/n)$ denote the size of the maximum matching on the graph $G(n,n,c/n)$ where $p = c/n$. Then a.a.s.,
\bea
\frac{\mu^*(n,n,c/n)}{n} \le 2-\frac{\gamma^*+\gamma_*+\gamma^* \gamma_*}{c} + o(1),
\eea
where $\gamma_*$ is the smallest root of the equation $x = c\exp(-ce^{-x})$ and $\gamma^* = c e^{-\gamma_*}$.
\end{thm}
\nin The bound is known to be tight for $c \le e$. Remarkably, this bound differs from the expression for the maximum matching size in the $G(n,p)$ model obtained in \cite{aronson98} by a factor of two.

\subsection{Oblivious algorithm.}
The oblivious algorithm performs a ``one shot'' trial for each ball $j$, where it attempts to match $j$ to a random neighbor. The algorithm is unaware of which bins are already matched, so an attempted match to an already occupied bin means that ball $j$ is dropped. This is shown in Algorithm \ref{algo:oblivious}. We use $N(j)$ to denote the set of neighboring bins of ball $j$.
\begin{algorithm}
\DontPrintSemicolon
\ForEach{$j \in J$}{
When ball $j$ arrives, let $N(j)$ denote the set of neighboring bins\;
\If{$|N(j)| >0$}{
Select a random neighboring bin $i \in N(j)$\;
\If{bin $i$ is unmatched}{
Match $j$ to $i$\;
}
\Else{
$j$ is dropped\;
}
}
}
\caption{{\sc oblivious}}
\label{algo:oblivious}
\end{algorithm}

Intuitively, the algorithm is expected to do well on extremely sparse graphs, specifically those where balls are likely to have at most one neighbor. On the other hand, the performance on dense graphs should be suboptimal since each ball has a variety of neighboring bins that are not utilized by the algorithm. The following two lemmas confirm this behavior.

\begin{lma}
For all valid functions $p(n)$ where $p(n) = o(1/n)$, the performance of \textsc{oblivous} on $G(n,n,p)$ with $p=p(n)$ satisfies $\mathcal{R}_{\mathrm{G}}(p(n)) = 1$.
\label{obliviouso}
\end{lma}
\begin{proof}
We consider the number of isolated edges in $G(n,n,p)$, where an edge $(i,j)$ is isolated if $N(i) = \{j\}$ and $N(j) = \{i\}$. For any bin and ball pair $(i,j)$, the probability of an isolated edge occurring between them is
\bea
\PP( (i,j)\mathrm{~is~isolated~edge}) = p(1-p)^{(2n-2)}.
\eea
Every isolated edge is matched by \textsc{oblivious}, so
\bea
\E[\mu_\mathrm{O}] \ge n^2p(1-p)^{(2n-2)},
\eea
where $\mu_\mr{O}$ is the size of the \textsc{oblivious} matching. The maximum matching size, denoted by $\mu^*$, is upper bounded in expectation by the expected number of edges, $n^2p$, so
\bea
\frac{\E[\mu_{\mathrm{O}}]}{\E[ \mu^*]} &\ge& (1-p)^{(2n-2)} = (e^{-o(\frac{1}{n})})^{(2n-2)}=  e^{-o(1)}  \rightarrow 1.
\eea
\qed
\end{proof}

\begin{lma}
For all valid functions $p(n)$ where $p(n) = \omega(1/n)$, the performance of \textsc{oblivous} on $G(n,n,p)$ with $p=p(n)$ satisfies $\mathcal{R}_{\mathrm{G}}(p(n)) = 1-1/e$.
\label{obliviousw}
\end{lma}
\begin{proof}
Fix a bin $i$. For each ball $j$, there is an attempted assignment of $j$ to $i$ with probability $1/n$ if $j$ has at least one neighbor. Thus,
\bea
\PP(\mathrm{no~attempted~match~of~}j\mathrm{~to~}i) = 1-\left( \frac{1-(1-p)^n}{n} \right).
\eea
Considering all balls and using $1-x \le e^{-x}$,
\bea
\PP(\mathrm{no~attempted~match~to~}i) &\le& \left(1-\left( \frac{1-e^{-pn}}{n} \right) \right)^n \n \\
&=& \left(1-\left( \frac{1-e^{-\omega(1)}}{n} \right) \right)^n \rightarrow \left (1-\frac{1}{n}\right)^n \rightarrow e^{-1}.
\eea
Accordingly, the expected number of bins matched by \textsc{oblivious} is at least $n(1-1/e)$. The maximum matching size is at most $n$, so the result follows.
\qed
\end{proof}

We are left to analyze the region $p(n)=c/n$, where we observe the transition from a performance ratio of $1$ to $1-1/e$. The proof of Theorem \ref{obliviousthm} is shown below; the proofs of Theorem \ref{roblivious} and Corollary \ref{rboundoblivious} then follow easily. Theorem \ref{obliviousthm} can be proved using a linearity of expectations argument; we instead use the approach below to introduce the differential equation method. ~\\

\nin \textit{Proof of Theorem \ref{obliviousthm}.}  
Let $Y(t)$ denote the number of occupied bins immediately before the $t^\mathrm{th}$ ball arrives, where $Y(1) = 0$. Slightly abusing notation, we use $Y(t)$ to denote both the random variable and instances of the random variable. Given $Y(t)$, the $t^\mathrm{th}$ ball is dropped if it is isolated or if its selected neighbor is already matched. Thus,
\bea
\PP(\mathrm{ball~} t \mathrm{~is~matched}|Y(t)) &=& (1-(1-p)^n) \left(1-\frac{Y(t)}{n} \right) \n \\
&=& \E[Y(t+1)-Y(t)|Y(t)].
\eea
Define the normalized random variable
\bea
Z(\tau) = \frac{Y(n\tau)}{n},\quad 0 \le \tau \le 1,
\eea
so that $Z(\tau)$ indicates the fraction of occupied bins after $n\tau$ of the arrivals have occurred. We have
\bea
\frac{\E[Z(\tau+1/n)-Z(\tau) | Z(\tau)]}{1/n} &=& (1-(1-p)^n)(1-Z(\tau)) \n \\
&=& \left(1-\left(1-\frac{c}{n}\right)^n\right)(1-Z(\tau)) \n \\
&=& (1-e^{-c})(1-Z(\tau)) + o(1).
\eea
As $n \rightarrow \infty$, we arrive at the differential equation
\bea
\frac{dz(\tau)}{d\tau} = (1-e^{-c})(1-z(\tau)),
\eea
where $z(\cdot)$ is deterministic and replaces $Z(\cdot)$. Integrating and using $z(0) = 0$ gives
\bea
z(\tau)=1-e^{(e^{-c}-1)\tau},\quad 0 \le \tau \le 1.
\eea

Applying Theorem \ref{wormaldthm}, we choose the domain $D$ defined by $- \epsilon < \tau < 1+\epsilon$ and  $- \epsilon < z < 1+\epsilon$, for $\epsilon>0$. Clearly we have $C_0 = 1$ and $\beta = 1$ by the nature of the matching process. Let $\lambda_1 = c/(en)$ for the trend hypothesis, which is satisfied according to Lemma \ref{ebound}. The Lipschitz hypothesis is satisfied with a Lipschitz constant $L = (1+\epsilon)(1-e^{-c})$. Setting $\tau=1$ and choosing $\lambda = cn^{-1/4}$, Theorem \ref{wormaldthm} (b) gives that with probability $\ds1-O(n^{1/4} e^{-c^3n^{1/4}})$,
\bea
\mu_{\mathrm{O}}(n,n,p) = n(1-e^{(e^{-c}-1)}) + O(n^{3/4}).
\eea
The extra $\epsilon$ portion of the domain ensures that $z(\tau)$ remains within distance $C \lambda$ of the boundary of $D$.

\qed

\nin \textit{Proof of Theorem \ref{roblivious}.}
The cases of $p(n) = o(1/n)$ and $p(n) = \omega(1/n)$ are given by Lemma \ref{obliviouso} and Lemma \ref{obliviousw}. For the regime $p = c/n$, define the normalized random variables 
\bea
\wt{\mu}_\mr{O}(n,n,c/n) &:=& \frac{\mu_\mr{O}(n,n,c/n)}{n}, \quad
\wt{\mu}^*(n,n,c/n) := \frac{\mu^*(n,n,c/n)}{n}.
\eea
By Theorem 1, $\wt{\mu}_\mr{O}$ converges in probability. Additionally, $\wt{\mu}_\mr{O}$ is bounded and thus uniformly integrable, so convergence in probability implies convergence in mean:
\bea
\lim_{n \rightarrow \infty} \E[\wt{\mu}_{\mr{O}}(n,n,c/n)] = 1-e^{(e^{-c}-1)}.
\eea
Similarly, $\wt{\mu}^*$ must satisfy
\bea
\lim_{n \rightarrow \infty} \E[\wt{\mu}^*(n,n,c/n)] \ge 2- \frac{\gamma^*+\gamma_* + \gamma^* \gamma_*}{c}.
\eea
\qed

\nin \textit{Proof of Corollary \ref{rboundoblivious}.}
The expected maximum matching size is no greater than the expected number of non-isolated vertices on one side of the graph, so $\mc{R}_\mr{O}(c/n) \ge \frac{1-e^{(e^{-c}-1)}}{1-e^{-c}} \ge 1-1/e$. \qed

\subsection{Greedy algorithm.}
The \textsc{greedy} algorithm is shown in Algorithm \ref{algo:greedy}. Upon the arrival of each ball, \textsc{greedy} assigns the ball to a randomly selected neighboring bin that is unmatched. If no such bin exists, the ball is dropped. Recall that the \textsc{ranking} algorithm instead picks an initial ranking of bins and always assigns each ball to the neighboring unmatched bin with highest rank. We begin by showing that \textsc{ranking} and \textsc{greedy} perform equivalently on $G(n,n,p)$. Let $U(j)$ denote the set of unmatched neighboring bins of ball $j$ when it arrives.
\begin{algorithm}
\DontPrintSemicolon
\ForEach{$j \in J$}{
When ball $j$ arrives, let $U(j)$ denote the set of unmatched neighboring bins\;
\If{$|U(j)| >0$}{
Match $j$ to a random bin $i \in U(j)$\;
}
\Else{
$j$ is dropped\;
}
}
\caption{{\sc greedy}}
\label{algo:greedy}
\end{algorithm}

\begin{lma}
The performance of \textsc{greedy} and \textsc{ranking} are equivalent on the graph $G(n,n,p)$.
\label{equiv}
\end{lma}
\begin{proof}
Consider the evolution of an instance of the \textsc{ranking} algorithm. Upon the arrival of each ball $j$, a match does not occur only if the bins $N(j)$ are all occupied or if $N(j) = \emptyset $. Immediately before the edges for $j$ are revealed, the probability of such an event depends only on the number occupied bins and is independent of the bin ranking, as edges are generated independently of the existing graph. It follows that the matches obtained by \textsc{greedy} and \textsc{ranking} are determined only as a function of the number of available bins.
\qed
\end{proof}


It is clear that on extremely sparse graphs, \textsc{greedy} should perform near optimally, as \textsc{oblivious} does. The \textsc{greedy} algorithm should also perform well on very dense graphs because many bins are available to each ball. We formalize this behavior with the following two lemmas.

\begin{lma}
For all valid functions $p(n)$ where $p(n) = o(1/n)$, the performance of \textsc{greedy} on $G(n,n,p)$ with $p=p(n)$ satisfies $\mathcal{R}_{\mathrm{G}}(p(n)) = 1$.
\end{lma}
\begin{proof}
The proof for Lemma \ref{obliviouso} holds. \qed
\end{proof}

\begin{lma}
For all valid functions $p(n)$ where $p(n) = \omega(1/n)$, the performance of \textsc{greedy} on $G(n,n,p)$ with $p=p(n)$ satisfies $\mathcal{R}_{\mathrm{G}}(p(n)) = 1$.
\end{lma}
\begin{proof}
We use a crude lower bound for the probability that each ball is matched. When ball $t$ arrives, at most $t-1$ bins are occupied, so
\bea
\PP(\mathrm{ball~} t \mathrm{~is~matched}) \ge 1-(1-p)^{(n-t+1)}.
\eea
Let $\mu_\mathrm{G}$ denote the matching size obtained by \textsc{greedy}. Then
\bea
\E[\mu_{\mathrm{G}}] \ge n-\sum_{t = 1}^n (1-p)^{(n-t+1)} = n - \frac{(1-(1-p)^n)(1-p)}{p}  \ge  n-\frac{1}{p}.
\eea
The maximum matching size is at most $n$, so
\bea
\frac{\E[ \mu_\mathrm{G}]}{\E[\mu^*]} &\ge& 1-\frac{1}{np(n)} = 1-\frac{1}{\omega(1)} \rightarrow 1.
\eea
\qed
\end{proof}

We now prove the result for the matching size obtained by \textsc{greedy} in the case $p=c/n$ (Theorem \ref{greedythm}) and follow with proofs of Theorem \ref{rgreedy} and Corollary \ref{rboundgreedy}.~\\

\nin\textit{Proof of Theorem \ref{greedythm}.} Again let $Y(t)$ denote the number of occupied bins immediately before the $t^\mathrm{th}$ arrival. Conditioning on $Y(t)$, the $t^\mathrm{th}$ ball cannot be assigned only if edges connecting the ball to the $n-Y(t)$ neighboring bins are not present. The probability of a match is then
\bea
\PP(\mathrm{ball~} t \mathrm{~is~matched}| Y_t)&=& 1-(1-p)^{n \left (1-\frac{Y(t)}{n} \right )} \n \\
&=& \E[Y(t+1)-Y(t)|Y(t)].
\eea
Normalizing as before, we obtain
\bea
\frac{\E[Z(\tau+1/n)-Z(\tau)|Z(\tau)]}{1/n} &=& 1- \left (1-\frac{c}{n} \right )^{n(1-Z(\tau))} \n \\
&=& 1-e^{-c(1-Z(\tau))} + o(1).
\eea
The corresponding differential equation for $n \rightarrow \infty$ is 
\bea
\frac{d z(\tau)}{d \tau} = 1-e^{-c(1-z(\tau))}.
\eea
By integration and the initial value for $z(\tau)$,
\be
z(\tau) = 1-\frac{\log{\left (1+e^{-c\tau}(e^c-1) \right )}}{c}, \quad 0 \le \tau \le 1.
\ee
The application of Wormald's theorem is the same as the proof for \textsc{oblivious} but with Lipschitz constant $L = ce^{c\epsilon}$.
\qed \\

\nin\textit{Proof of Theorem \ref{rgreedy}}. The proof follows the same approach as the proof of Theorem  \ref{roblivious}.
\qed

~\\
\nin\textit{Proof of Corollary \ref{rboundgreedy}.}
The approach for Corollary \ref{rboundoblivious} can be used to show that the property holds for small $c$ values (e.g. $c<1/2$) and large $c$ values (e.g. $c>5$). For the remaining region, the exact expression for the performance ratio can be minimized numerically; the minimum is obtained at $c^* = 3.1685009$ and $\mc{R}_\mr{G}(c^*) = 0.8370875$.
\qed

\section{Vertex-Weighted Matching}
In this section we modify the setting so that each bin $i$ has a \textit{rank}, $r_i \in \{1,2,\ldots, m \}$ for some constant $m$, and that matching balls to lower ranked bins is preferable. This is equivalent to allowing each bin to have one of $m$ distinct weights $w_i \in \R$, where a larger weighted bin has a lower rank.
The \textsc{vertex-weighted greedy} algorithm matches each ball to its neighboring bin with smallest rank, as shown in Algorithm \ref{vertexweightedmatching}.

\label{vertexweightedmatching}
\begin{algorithm}
\DontPrintSemicolon
\ForEach{$j \in J$}{
When ball $j$ arrives, let $U(j)$ denote the set of unmatched neighboring bins\;
\If{$|U(j)| >0$}{
Match $j$ to a neighboring bin $\ds i \in \argmin_{i \in N(j)} r_i$\;
}\Else{
$j$ is dropped\;

}
}
\caption{{\sc vertex-weighted greedy}}
\label{algo:weighted}
\end{algorithm}
\nin Let $n_r$ denote the total number of bins with rank $r$, where $\sum_r n_r = n$, and let $g_r := n_r/n$. The result of this section is as follows.

\begin{thm}
\label{weightedthm}
Let $\mu^r_{\mathrm{W}}(n,n,c/n)$ denote the number of bins with rank $r$  that are matched by \textsc{vertex-weighted greedy} on the graph $G(n,n,p)$, where $p=c/n$ and $c>0$ is a constant. Then a.a.s.,
\bea
\frac{\mu^r_{\mathrm{W}}(n,n,c/n)}{n} = g_r - \frac{1}{c} \log \left ( \frac{1+e^{-c\tau}(e^{c\sum_{k=1}^{r} g_k} -1)}{1+e^{-c\tau}(e^{c\sum_{k=1}^{r-1} g_k} -1)} \right ) + o(1), \quad 1 \le r \le m.
\eea
\end{thm}

\begin{proof}
%
We use the notation $Y_r(t),~1\le r \le m$, to denote the number of rank $r$ bins that are occupied immediately prior to the arrival of ball $t$. Consider the probability that a given ball $t$ is matched to a rank two bin. This occurs if edges connecting to the $n_1 - Y_1(t)$ are not present and there is at least one neighboring bin with rank two, so
\bea
\E[ Y_2(t+1) - Y_2(t) | \mathbf{Y}(t)] = \left ( (1-p)^{n \left (g_1 - \frac{Y_1(t)}{n} \right )} \right ) \left ( 1-(1-p)^{n\left (g_2 - \frac{Y_2(t)}{n}\right )}
 \right ),
\eea
where $\mathbf{Y}(t) = (Y_1(t), Y_2(t), \ldots Y_a(t))$. Generalizing, for $1 \le r \le m$,
\bea
\E[ Y_r(t+1) - Y_r(t) | \mathbf{Y}(t)] &=& \left ( 1-(1-p)^{n \left (g_r - \frac{Y_r(t)}{n} \right )} \right ) \prod_{k =1}^{r-1} \left ( (1-p)^{n\left (g_k - \frac{Y_k(t)}{n}\right)} \right ).
\eea
After normalizing and substituting $p=c/n$,
\bea
\frac{\E[Z_r(\tau+1/n)-Z_r(\tau)|\mathbf{Z}(\tau)]}{1/n}&=& \left(1-e^{-c (g_r- Z_r(\tau))} \right) e^{-c\sum_{k=1}^{r-1}\left (g_k - Z_k(\tau) \right)} + o(1).
\eea
We arrive at the following system of differential equations for $n \rightarrow \infty:$
\bea
\frac{dz_r(\tau)}{d \tau} = \left(1-e^{-c (g_r - z_r(\tau))} \right) e^{-c\sum_{k=1}^{r-1}\left (g_k - z_k(\tau) \right)}, ~\quad 1 \le r \le m.
\label{diffeqsys}
\eea
It can be verified that the solution to the system of differential equations with initial conditions $z_r(0) = g_r$, $1 \le r \le m$, is
\bea
z_r(\tau) = g_r - \frac{1}{c} \log \left ( \frac{1+e^{-c\tau}(e^{c\sum_{k=1}^{r} g_k} -1)}{1+e^{-c\tau}(e^{c\sum_{k=1}^{r-1} g_k} -1)} \right ), \quad 0 \le \tau \le 1,
\eea
for $1 \le r \le m$.
The application of Wormald's theorem is the same as the application of \textsc{oblivious} and \textsc{greedy}, but with Lipschitz constant $L = (a-1)(1-e^{-c(1+\epsilon)})ce^{(a-1)c\epsilon}+ce^{ac\epsilon}$.
\qed
\end{proof}

\section{Conclusion}
\label{conclusion}
We have determined lower bounds on the performance of greedy matching algorithms on $G(n,n,p(n))$ that hold across all monotone functions $p(n)$, and shown that a clear phase transition occurs at $p(n) \sim 1/n$. In particular, $p(n) \sim 1/n$ is where the performance ratio of the \textsc{oblivious} algorithm transitions from $1$ to $1-1/e$ and the performance ratio of \textsc{greedy} passes through its global minimum. 

The greedy algorithms perform relatively well: it is interesting that \textsc{oblivious}, which is clearly inferior to \textsc{ranking}, has a minimum performance ratio that is equal to the best possible worst-case competitive ratio of any online matching algorithm (i.e. that of \textsc{ranking}). Likewise, the lower bound of 0.837 on the performance ratio of \textsc{greedy} is surprisingly high. Our analysis is tight inasmuch as the bound on the expected maximum matching size is tight (specifically Theorem \ref{matchopt}); we conjecture that this theorem is in fact tight for all $c>0$. Even if this is not the case, the lower bound on the performance of \textsc{greedy} cannot be any greater than $0.845$, where the bound is tight for $c=e$.

There appears to be a close relationship between matching behavior on $G(n,n,p)$ and $G(n,p)$, where results differ by a mere factor of two. As we have mentioned, this relationship occurs with our asymptotic result for \textsc{greedy} (Theorem \ref{greedythm}) and the result of \textsc{modified greedy} given in \cite{dyer93}. If Theorem \ref{matchopt} is in fact tight for all $c>0$, the same can be said regarding a result given in \cite{aronson98}. A similar observation was also made by Frieze with respect to differential equations for the Karp-Sipser algorithm on bipartite and ordinary graphs \cite{frieze05}. Note that it is not sufficient to simply argue that these properties hold for the mere fact that bipartite graphs have twice as many vertices; there are indeed important structural differences between ordinary and bipartite graphs.

A variety of possibilities exist for further research. We have allowed the existence of isolated vertices in our analysis, which is unlikely to be realistic for many applications. This could be resolved by imposing a restriction on the minimum degree of vertices, as was done in \cite{frieze04}. Unbalanced bipartite graphs (i.e. with more vertices on one side) are likely to be encountered in practice -- our approach can be used in this situation, but less is known about expected maximum matching size here. Stepping away from the online scenario, one could consider multiple pass algorithms, which are highly relevant to streaming models.

\bibliographystyle{siam}
\bibliography{greedy_bip}

\end{document}